\documentclass[11pt]{amsart} 
\usepackage{a4wide}
\usepackage{amsmath,amssymb,graphicx}
\usepackage{mathrsfs}
 \usepackage{mathtools}
 \usepackage{float}
\usepackage{dsfont}
\usepackage{comment}
\usepackage[T1]{fontenc}
\usepackage[utf8]{inputenc} 
\usepackage{enumitem}
\usepackage{subcaption}
\usepackage{hyperref}
\usepackage{tikz}
\usetikzlibrary{decorations.pathreplacing,decorations.markings,positioning,arrows.meta}
\usetikzlibrary{matrix}
\usepackage{stmaryrd}
\usepackage{xcolor}
\def\di{\displaystyle}
\def\N{\mathbb{N}}
\def\R{\mathbb{R}}

\def\di{\displaystyle}
\def\T{\mathbf{T}}

\newtheorem{theorem}{Theorem}[section]
\newtheorem{lemma}[theorem]{Lemma}
\newtheorem{definition}[theorem]{Definition}

\usepackage{graphicx} 

\newcommand{\ack}{\section*{Acknowledgements}}

\begin{document}


\title{From fractal R-L ladder networks to the diffusion equation}

\author{Jacky Cresson$^1$, Anna Szafra\'{n}ska$^2$}

\begin{abstract}
We give a self-contained presentation of fractal R-L ladder networks as well as a detailed computation of the admittance of these systems. We also discuss the conditions under which such systems display a fractional behavior. Finally, we give a full discussion of the connection existing between fractal R-L network and the diffusion equation. 
\end{abstract}


\maketitle

$^1$ Laboratoire de mathématiques et leurs applications, UMR CNRS 5142, Université de Pau et des Pays de l'Adour-E2S, France

$^2$  Institute of Applied Mathematics, Gdańsk University of Technology, G. Narutowicz Street 11/12, 80-233 Gdańsk, Poland
\tableofcontents

\section{Introduction}
\label{sec1}

We are interested in a dynamical system $S$ which is modeled by a functional relationship between system input $e(t)$ and system output $u(t)$.\\
\begin{figure}[ht]
  \centering
   \includegraphics[width = 0.5\textwidth]{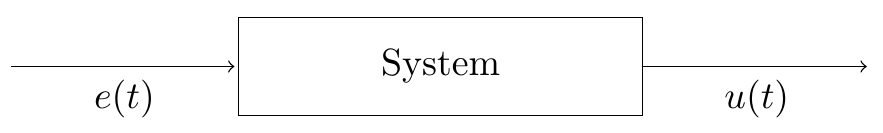}
\end{figure}\\
Let $E$ and $U$ denote the image under the Fourier transform $\mathscr{F}$ of the functions $e$ and $u$, respectively, i.e. 
\begin{equation} \label{laplace}
E(\omega) = \mathscr{F}[e] (\omega),\;\;\;\; U(\omega) = \mathscr{F}[u](\omega).
\end{equation}
The system $S$ can then be described using a {\bf transfer function} $H$ connecting the functions $E$ and $U$, precisely
\begin{equation} \label{transfer}
U(\omega) = H(\omega)E(\omega).
\end{equation}
In this article, we study a special class of electronic systems called {\bf ladder networks}. Ladder networks have been used in a variety of situations like modeling of electric machines \cite{riu}, respiratory system \cite{clara}, hydrogen storage \cite{sabatier}, etc. Many examples can be found in the book of A. Oustaloup \cite{oustaloup}. \\

We are interested in ladder-network which possess a {\bf fractional behavior}, meaning that the transfer function $H$ behaves as 
\begin{equation} \label{transferBehaviour}
H(\omega ) \sim_{\omega \rightarrow +\infty} \omega^{\nu} ,
\end{equation}
for a given constant $\nu>0$.\\

In this paper, we focus on ladder-network of the form presented on Figure \ref{RL},\\
\begin{center}
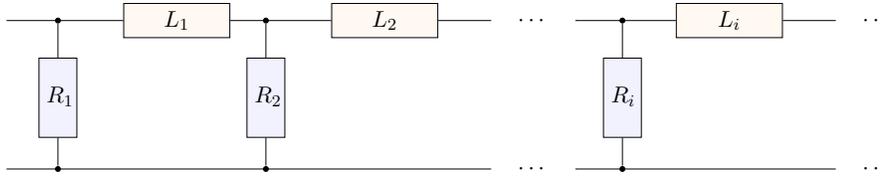
\begin{figure}[ht]
    \centering
    \resizebox{0.8\linewidth}{!}{	\begin{tikzpicture}
		[-,auto,thin,
			boxR/.style={rectangle,text width=1em,fill=blue!5,draw,align=center},
			boxL/.style={rectangle,text width=4em,fill=orange!5,draw,align=center},
			mid arrow/.style={draw, postaction={decorate},
				decoration={
					markings, mark=at position 0.75 with {\arrow[scale=2]{>}}}},
			end arrow/.style={draw, postaction={decorate},
				decoration={
					markings, mark=at position 1 with {\arrow[scale=2]{>}}}}	
			]
		\node[] (empty00) {};
		\node[circle,fill,inner sep=1pt] (empty0R1) [right of = empty00,node distance = 1cm] {};
		\node[boxR] (R1) [below of = empty0R1,node distance = 1.3cm] {\phantom{x}\\$R_1$\\ \phantom{x}};
		\node[boxL] (L1) [right of = empty0R1,node distance = 2cm] {$L_1$};	
		\node[circle,fill,inner sep=1pt] (empty0R2) [right of = L1,node distance = 1.5cm,text width=0.001] {};
		\node[boxR] (R2) [below of = empty0R2,node distance = 1.3cm] {\phantom{x}\\$R_2$\\ \phantom{x}}; 
		\node[boxL] (L2) [right of = empty0R2,node distance = 2cm] {$L_2$};	
		\node[] (dot00) [right of = L2,align = center,node distance = 2.5cm,text width=3em] {$\cdots$};
		\node[circle,fill,inner sep=1pt] (empty0Ri) [right of = dot00,node distance = 1.5cm] {};
		\node[boxL] (Li) [right of = empty0Ri ,node distance = 1.8cm] {$L_i$};
		\node[boxR] (Ri) [below of = empty0Ri,node distance = 1.3cm] {\phantom{x}\\$R_i$\\ \phantom{x}}; 
		\node[] (dot01) [right of = Li,align=center,node distance = 2.5cm,text width=3em] {$\cdots$};
		
		\node[] (empty10) [below of = empty00,node distance = 2.5cm] {};
		\node[circle,fill,inner sep=1pt] (empty1R1) [below of = empty0R1,node distance = 2.5cm] {};
		\node[circle,fill,inner sep=1pt] (empty1R2) [below of = empty0R2,node distance = 2.5cm] {};
		\node[] (dot11) [below of = dot00,node distance = 2.5cm,align=center,text width=3em] {$\cdots$};
		\node[circle,fill,inner sep=1pt] (empty1Ri) [below of = empty0Ri,node distance = 2.5cm] {};
		\node[] (dot12) [below of = dot01,align=center,node distance = 2.5cm,text width=3em] {$\cdots$};
		
					
		\draw[] (empty00) -- (L1) -- (L2) -- (dot00) -- (Li) -- (dot01);
		\draw[] (empty0R1) -- (R1) -- (empty1R1);
		\draw[] (empty0R2) -- (R2) -- (empty1R2);
		\draw[] (empty0Ri) -- (Ri) -- (empty1Ri);
		\draw[] (empty10) -- (dot11) -- (dot12);

	\end{tikzpicture}}
    \caption{R-L ladder-network} \label{RL}
\end{figure}
\end{center}
where $R_i$ and $L_i$ are resistance and inertance respectively. Imposing scaling relations on the resistance and inertance, we obtain {\bf fractal R-L ladder networks}. They were defined by A. Oustaloup \cite{oustaloup} (see also \cite{riu},p.38, figure III.10).\\

We have not find a complete derivation of the expression given in (\cite{oustaloup}, \cite{riu},equation (III-38)) for the transfer function nor of the functional relation satisfied by this function.\\

In this paper, we give a self-contained characterization of the transfer function associated to R-L ladder network and how this function specialise in the case of fractal R-L ladder networks.

We study under which conditions such fractal R-L ladder networks produce fractional behavior. In particular, we give complete proofs for results presented in (\cite{oustaloup},\cite{riu}, III.42) as well as details about the connection with the diffusion equation. Our presentation follows the general strategy proposed by J. Sabatier and al. in \cite{sabatier} interpreting some electronic devices as discretization of some partial differential equations and in particular as particular diffusion equations. \\

The paper is organised as follows:\\

In Section \ref{sec2}, we introduce the mathematical framework about trees and forest allowing us to encode the structure of the electronic diagram describing the R-L ladder networks. We derive in Section \ref{sectionrecursive}, the explicit expression of the transfer function for a R-L ladder network which makes use of continued fraction expansions. 

Section \ref{sectioncomputation} specialise the previous result in the case of fractal R-L ladder networks. In particular, we prove that the transfer function satisfies a functional relation.

In Section \ref{sectionfractionalbehavior}, we discuss under which conditions a transfer function can exhibit a fractional behavior.

Finally Section \ref{sectiondiffusion} describes in full details the connection between the diffusion equation and fractal R-L ladder networks.

We finish by some perspectives of this work.


\section{Electronic diagrams and decorated forest}
\label{sec2}

Every {\bf electronic diagram} is constructed using two basic configurations connecting electronic elements called in {\bf series} and {\bf parallel} and classically represented as follows \begin{itemize}
    \item in series
    \begin{figure}[ht]
    \centering
    \resizebox{0.3\linewidth}{!}{	\begin{tikzpicture}
		[-,auto,thin,
			boxR/.style={rectangle,text width=3em,fill=blue!5,draw,align=center},
			boxC/.style={rectangle,text width=1em,fill=orange!5,draw,align=center},
			mid arrow/.style={draw, postaction={decorate},
				decoration={
					markings, mark=at position 0.75 with {\arrow[scale=2]{>}}}},
			end arrow/.style={draw, postaction={decorate},
				decoration={
					markings, mark=at position 1 with {\arrow[scale=2]{>}}}}	
			]

		\node[] (empty0) {};
		\node[boxR] (S1) [right of = empty0,node distance = 1.5cm] {\phantom{b}\hspace{-0.5em}$a$};
		\node[boxR] (S2) [right of = S1,node distance = 2.2cm] {$b$};
		\node[] (empty1) [right of = S2,node distance = 1.5cm]{};

		\draw[] (empty0) -- (S1) node[above=1mm,text width=7em] {};
		\draw[] (S1) -- (S2) -- (empty1) {};
		
	\end{tikzpicture}}
\end{figure}
\item in parallel 
    \begin{figure}[ht]
    \centering
    \resizebox{0.23\linewidth}{!}{	\begin{tikzpicture}
		[-,auto,thin,
			boxR/.style={rectangle,text width=3em,fill=blue!5,draw,align=center},
			boxC/.style={rectangle,text width=1em,fill=orange!5,draw,align=center},
			connection/.style={inner sep=0,outer sep=0},
			mid arrow/.style={draw, postaction={decorate},
				decoration={
					markings, mark=at position 0.75 with {\arrow[scale=2]{>}}}},
			end arrow/.style={draw, postaction={decorate},
				decoration={
					markings, mark=at position 1 with {\arrow[scale=2]{>}}}}	
			]

		\node[connection] (0) [text width = 0.0001em] {};
		\node[connection] (1) [right of = 0,node distance = 0.8cm,text width = 0.0001em] {};
		\node[connection] (11) [above of = 1, node distance = 0.5cm] {};
		\node[connection] (12) [below of = 1, node distance = 0.5cm]{};
		
		\node[boxR] (S1) [right of = 11,node distance = 1.2cm] {\phantom{b}\hspace{-0.5em}$a$};
		\node[boxR] (S2) [right of = 12,node distance = 1.2cm] {$b$};
		
		\node[connection] (31) [right of = S1, node distance = 1.2cm]{};
		\node[connection] (32) [right of = S2, node distance = 1.2cm]{};
		\node[connection] (3) [below of = 31, node distance = 0.5cm]{};
		\node[connection] (4) [right of = 3, node distance = 0.8cm]{};

					
		\draw (0) -- (1) ;
		\draw (11) -- (12) -- (S2) -- (32);		
		\draw (11) -- (S1) -- (31) -- (32);
		\draw (3) -- (4) {};

	\end{tikzpicture}}
\end{figure}
\end{itemize}
where $a$ and $b$ are two electronic components.\\

The idea to represent electronic circuits as {\bf graphs} goes back to G. Kirchhoff in 1847 \cite{kirchhoff}. We use the following graphs for the series and parallel diagrams respectively:

$$
\begin{tikzpicture}
\node[circle,fill,inner sep=2pt] (tree1) [node distance = 2cm]{};
		\node[] (texta) [below of = tree1,node distance = 0.45cm]{$a$};
		\node[circle,fill,inner sep=2pt] (tree2) [right of = tree1,node distance = 1cm]{};
		\node[] (textb) [below of = tree2,node distance = 0.4cm]{$b$};
		\draw[] (tree1) -- (tree2);	
\end{tikzpicture}
\qquad \ \mbox{\rm and}\qquad \ 
\begin{tikzpicture}
\node[circle,fill,inner sep=2pt] (tree1) [node distance = 2cm]{};
		\node[] (texta) [below of = tree1,node distance = 0.45cm]{$a$};
		\node[circle,fill,inner sep=2pt] (tree2) [right of = tree1,node distance = 1cm]{};
		\node[] (textb) [below of = tree2,node distance = 0.4cm]{$b$};
\end{tikzpicture}  
$$

An alternatively more algebraic way is to encode the previous graphs as $ab$ when $a$ and $b$ are in series and $a\otimes b$ when $a$ and $b$ are in parallel.\\

As an example, we consider the following family of {\bf recursive diagrams} denoted by $D_n$, $n\geq 1$, represented as it is shown in Figure \ref{DDD}.
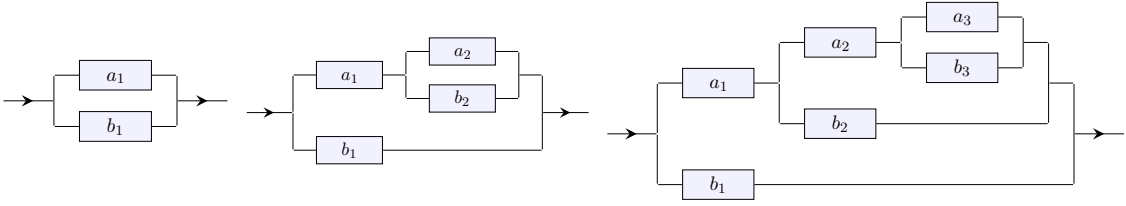
\begin{figure}[!ht]
\centering
\phantom{xx}
\begin{minipage}{0.2\textwidth}
    \centering
    \resizebox{1\linewidth}{!}{	\begin{tikzpicture}
		[-,auto,thin,
			boxR/.style={rectangle,text width=3em,fill=blue!5,draw,align=center},
			boxC/.style={rectangle,text width=1em,fill=orange!5,draw,align=center},
			connection/.style={inner sep=0,outer sep=0},
			mid arrow/.style={draw, postaction={decorate},
				decoration={
					markings, mark=at position 0.6 with {\arrow[scale=2]{>}}}},
			end arrow/.style={draw, postaction={decorate},
				decoration={
					markings, mark=at position 1 with {\arrow[scale=2]{>}}}}	
			]

		\node[connection] (0) [text width = 0.0001em] {};
		\node[connection] (1) [right of = 0,node distance = 1cm,text width = 0.0001em] {};
		\node[connection] (11) [above of = 1, node distance = 0.5cm] {};
		\node[connection] (12) [below of = 1, node distance = 0.5cm]{};
		
		\node[boxR] (S1) [right of = 11,node distance = 1.2cm] {\phantom{b}\hspace{-0.5em}$a_1$};
		\node[boxR] (S2) [right of = 12,node distance = 1.2cm] {$b_1$};
		
		\node[connection] (31) [right of = S1, node distance = 1.2cm]{};
		\node[connection] (32) [right of = S2, node distance = 1.2cm]{};
		\node[connection] (3) [below of = 31, node distance = 0.5cm]{};
		\node[connection] (4) [right of = 3, node distance = 1cm]{};

					
		\draw[mid arrow,>=stealth] (0) -- (1) ;
		\draw (11) -- (12) -- (S2) -- (32);		
		\draw (11) -- (S1) -- (31) -- (32);
		\draw[mid arrow,>=stealth] (3) -- (4) {};

	\end{tikzpicture}}
\end{minipage}
\begin{minipage}{0.3\textwidth}
    \centering
    \resizebox{1\linewidth}{!}{	\begin{tikzpicture}
		[-,auto,thin,
			boxR/.style={rectangle,text width=3em,fill=blue!5,draw,align=center},
			boxC/.style={rectangle,text width=1em,fill=orange!5,draw,align=center},
			connection/.style={inner sep=0,outer sep=0},
			mid arrow/.style={draw, postaction={decorate},
				decoration={
					markings, mark=at position 0.6 with {\arrow[scale=2]{>}}}},
			end arrow/.style={draw, postaction={decorate},
				decoration={
					markings, mark=at position 1 with {\arrow[scale=2]{>}}}}	
			]

		\node[connection] (0) [text width = 0.0001em] {};
		\node[connection] (1) [right of = 0,node distance = 1cm] {};
		\node[connection] (11) [above of = 1, node distance = 0.8cm] {};
		\node[connection] (12) [below of = 1, node distance = 0.8cm]{};
		
		\node[boxR] (S1-11) [right of = 11,node distance = 1.2cm] {\phantom{b}\hspace{-0.5em}$a_1$};
		\node[boxR] (S1-12) [right of = 12,node distance = 1.2cm] {$b_1$};
		
		\node[connection] (2) [right of = S1-11, node distance = 1.2cm]{};
		\node[connection] (21) [above of = 2, node distance = 0.5cm] {};
		\node[connection] (22) [below of = 2, node distance = 0.5cm]{};
		
		\node[boxR] (S2-21) [right of = 21,node distance = 1.2cm] {\phantom{b}\hspace{-0.5em}$a_2$};
		\node[boxR] (S2-22) [right of = 22,node distance = 1.2cm] {$b_2$};
		
		\node[connection] (31) [right of = S2-21, node distance = 1.2cm]{};
		\node[connection] (32) [right of = S2-22, node distance = 1.2cm]{};
		\node[connection] (3) [below of = 31, node distance = 0.5cm]{};
		
		\node[connection] (41) [right of = 3, node distance = 0.5cm]{};
		\node[connection] (42) [below of = 41, node distance = 1.6cm]{};
		\node[connection] (4) [below of = 41, node distance = 0.8cm]{};
		
		\node[connection] (5) [right of = 4, node distance = 1cm]{};

					
		\draw[mid arrow,>=stealth] (0) -- (1);
		\draw (1)-- (12) -- (S1-12) -- (42) --(41) -- (3);		
		\draw (1) -- (11) -- (S1-11) -- (2) -- (21) -- (S2-21) -- (31) -- (32) -- (S2-22) -- (22) -- (2);
		\draw[mid arrow,>=stealth] (4) -- (5) {};

	\end{tikzpicture}} 
\end{minipage}
\begin{minipage}{0.45\textwidth}
    \centering
    \resizebox{1\linewidth}{!}{	\begin{tikzpicture}
		[-,auto,thin,
			boxR/.style={rectangle,text width=3em,fill=blue!5,draw,align=center},
			boxC/.style={rectangle,text width=1em,fill=orange!5,draw,align=center},
			connection/.style={inner sep=0,outer sep=0},
			mid arrow/.style={draw, postaction={decorate},
				decoration={
					markings, mark=at position 0.6 with {\arrow[scale=2]{>}}}},
			end arrow/.style={draw, postaction={decorate},
				decoration={
					markings, mark=at position 1 with {\arrow[scale=2]{>}}}}	
			]

		\node[connection] (0) [text width = 0.0001em] {};
		\node[connection] (1) [right of = 0,node distance = 1cm] {};
		\node[connection] (11) [above of = 1, node distance = 1cm] {};
		\node[connection] (12) [below of = 1, node distance = 1cm]{};
		
		\node[boxR] (S1-11) [right of = 11,node distance = 1.2cm] {\phantom{b}\hspace{-0.5em}$a_1$};
		\node[boxR] (S1-12) [right of = 12,node distance = 1.2cm] {$b_1$};
		
		\node[connection] (2) [right of = S1-11, node distance = 1.2cm]{};
		\node[connection] (21) [above of = 2, node distance = 0.8cm] {};
		\node[connection] (22) [below of = 2, node distance = 0.8cm]{};
		
		\node[boxR] (S2-21) [right of = 21,node distance = 1.2cm] {\phantom{b}\hspace{-0.5em}$a_2$};
		\node[boxR] (S2-22) [right of = 22,node distance = 1.2cm] {$b_2$};
		
		\node[connection] (3) [right of = S2-21, node distance = 1.2cm]{};
		\node[connection] (31) [above of = 3, node distance = 0.5cm] {};
		\node[connection] (32) [below of = 3, node distance = 0.5cm]{};
		
		\node[boxR] (S3-31) [right of = 31,node distance = 1.2cm] {\phantom{b}\hspace{-0.5em}$a_3$};
		\node[boxR] (S3-32) [right of = 32,node distance = 1.2cm] {$b_3$};
		
		\node[connection] (41) [right of = S3-31, node distance = 1.2cm]{};
		\node[connection] (42) [right of = S3-32, node distance = 1.2cm]{};
		\node[connection] (4) [below of = 41, node distance = 0.5cm]{};
		
		\node[connection] (51) [right of = 4, node distance = 0.5cm]{};
		\node[connection] (52) [below of = 51, node distance = 1.6cm]{};
		\node[connection] (5) [below of = 51, node distance = 0.8cm]{};
		
		\node[connection] (61) [right of = 5, node distance = 0.5cm]{};
		\node[connection] (62) [below of = 61, node distance = 2cm]{};
		\node[connection] (6) [below of = 61, node distance = 1cm]{};
		
		\node[connection] (7) [right of = 6, node distance = 1cm]{};			
					
					
		\draw[mid arrow,>=stealth] (0) -- (1);
		\draw (1)-- (12) -- (S1-12) -- (62) --(61) -- (5);		
		\draw (1) -- (11) -- (S1-11) -- (2) -- (22) -- (S2-22) -- (52) -- (51) -- (4);
		\draw (2) -- (21) -- (S2-21) -- (3) -- (31) -- (S3-31) -- (41) -- (42) -- (S3-32) -- (32) -- (3); 
		\draw[mid arrow,>=stealth] (6) -- (7) {};

	\end{tikzpicture}}
\end{minipage}
\caption{$D_1, D_2$ and $D_3$ diagrams.} \label{DDD}
\end{figure}


It must be noted that recursive diagrams contain as a special case R-L ladder networks.\\

Using the previous notation, we obtain the encoding for $D_1$, $D_2$ and $D_3$ presented in the Figure \ref{TTT}.

\begin{figure}[!ht]
    \resizebox{0.75\linewidth}{!}{	\begin{tikzpicture}
		[-,auto,thin,
			boxR/.style={rectangle,text width=4em,fill=blue!5,draw,align=center},
			boxC/.style={rectangle,text width=1em,fill=orange!5,draw,align=center},
			connection/.style={inner sep=0,outer sep=0},
			mid arrow/.style={draw, postaction={decorate},
				decoration={
					markings, mark=at position 0.75 with {\arrow[scale=2]{>}}}},
			end arrow/.style={draw, postaction={decorate},
				decoration={
					markings, mark=at position 1 with {\arrow[scale=2]{>}}}}	
			]

		\node[connection] (0) [text width = 0.0001em] {};
		
		
		\node[circle,fill,inner sep=2pt] (tree11) [right of = 0,node distance = 2cm]{};
		\node[] (texta1) [below of = tree11,node distance = 0.45cm]{$a_1$};
		\node[circle,fill,inner sep=2pt] (tree12) [right of = tree11,node distance = 1cm]{};
		\node[] (textb1) [below of = tree12,node distance = 0.4cm]{$b_1$};
			
		
		\node[circle,fill,inner sep=2pt] (tree21) [right of = tree12,node distance = 4cm]{};
		\node[] (texta12) [below of = tree21,node distance = 0.45cm]{$a_1$};
		\node[circle,fill,inner sep=2pt] (tree22) [right of = tree21,node distance = 1cm]{};
		\node[] (textb12) [below of = tree22,node distance = 0.4cm]{$b_1$};
		
		\node[connection] (T2-1) [above of = tree21,node distance = 1cm]{};
		\node[circle,fill,inner sep=2pt] (tree23) [left of = T2-1,node distance = 0.5cm]{};
		\node[] (texta21) [left of = tree23,node distance = 0.45cm]{$a_2$};
		\node[circle,fill,inner sep=2pt] (tree24) [right of = T2-1,node distance = 0.5cm]{};
		\node[] (texta22) [right of = tree24,node distance = 0.45cm]{$b_2$};
		
		
		\draw (tree21) -- (tree23);
		\draw (tree21) -- (tree24);
		
		
		\node[circle,fill,inner sep=2pt] (tree31) [right of = tree22,node distance = 4cm]{};
		\node[] (texta13) [below of = tree31,node distance = 0.45cm]{$a_1$};
		\node[circle,fill,inner sep=2pt] (tree32) [right of = tree31,node distance = 1cm]{};
		\node[] (textb13) [below of = tree32,node distance = 0.4cm]{$b_1$};
		
		\node[connection] (T3-1) [above of = tree31,node distance = 1cm]{};
		\node[circle,fill,inner sep=2pt] (tree33) [left of = T3-1,node distance = 0.5cm]{};
		\node[] (texta211) [left of = tree33,node distance = 0.45cm]{$a_2$};
		\node[circle,fill,inner sep=2pt] (tree34) [right of = T3-1,node distance = 0.5cm]{};
		\node[] (texta222) [right of = tree34,node distance = 0.45cm]{$b_2$};
		
		\node[connection] (T3-2) [above of = tree33,node distance = 1cm]{};
		\node[circle,fill,inner sep=2pt] (tree35) [left of = T3-2,node distance = 0.5cm]{};
		\node[] (texta213) [left of = tree35,node distance = 0.45cm]{$a_3$};
		\node[circle,fill,inner sep=2pt] (tree36) [right of = T3-2,node distance = 0.5cm]{};
		\node[] (texta224) [right of = tree36,node distance = 0.45cm]{$b_3$};
		
		
		\draw (tree31) -- (tree33);
		\draw (tree31) -- (tree34);
		\draw (tree33) -- (tree35);
		\draw (tree33) -- (tree36);
					
	\end{tikzpicture}}
\caption{T1, T2 and T3} \label{TTT}
\end{figure}
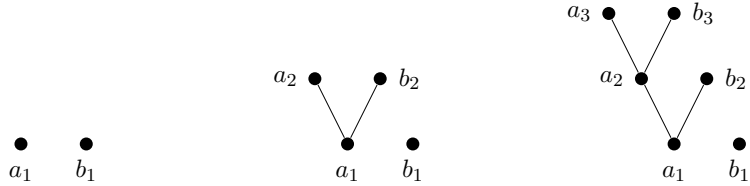

The special structure of recursive diagrams produces particular objects called {\bf forest} for the graph notation and {\bf words} for the algebraic one. Precisely, we have:

\begin{definition}
A rooted tree is a finite, connected graph without cycles, with a special vertex called the root.
\end{definition}

We denote by $\mathbf{T}$ the set of rooted trees. The number of vertices of rooted tree we call the \textbf{weight} of this tree. Then we denote by $\T(n)$ the set of rooted trees of weight $n$. As an example, we have 
$$
\T(1) = \{\bullet\},\;\; 
\T(2) = \Big\{\begin{tikzpicture}[baseline=-0.65ex]
\node[circle,fill,inner sep=1.5pt](t1) at (0,0.2){};
\node[circle,fill,inner sep=1.5pt](t2) at (0,-0.2){};
\draw (t1)--(t2);
\end{tikzpicture}\Big\},\;\;
\T(3) = \Big\{\begin{tikzpicture}[baseline=-0.65ex]
\node[circle,fill,inner sep=1.5pt](t1) at (-0.3,0){};
\node[circle,fill,inner sep=1.5pt](t2) at (0,0){};
\node[circle,fill,inner sep=1.5pt](t3) at (0.3,0){};
\end{tikzpicture},\;
\begin{tikzpicture}[baseline=-0.65ex]
\node[circle,fill,inner sep=1.5pt](t1) at (-0.2,0.2){};
\node[circle,fill,inner sep=1.5pt](t2) at (0.2,0.2){};
\node[circle,fill,inner sep=1.5pt](t3) at (0,-0.2){};
\draw (t1)--(t3)--(t2);
\end{tikzpicture}
\Big\}, \;\;\ldots
$$

As already seen in the representation of the diagram $D_n$, we have not rooted trees but a finite collection of rooted trees called forests. Precisely:

\begin{definition}
    A rooted forest is a graph such that each connected component is rooted tree. 
\end{definition}

The set of forest is denoted by $\textbf{F}$. The set of rooted forest of weight $n$ is denoted by $\textbf{F}(n)$
$$
\textbf{F}(1) = \{\bullet\},\;\;
\textbf{F}(2) = \Big\{\begin{tikzpicture}[baseline=-0.65ex]
\node[circle,fill,inner sep=1.5pt](t1) at (0,0.2){};
\node[circle,fill,inner sep=1.5pt](t2) at (0,-0.2){};
\draw (t1)--(t2);
\end{tikzpicture},\; 
\begin{tikzpicture}[baseline=-0.65ex]
\node[circle,fill,inner sep=1.5pt](t1) at (-0.15,0){};
\node[circle,fill,inner sep=1.5pt](t3) at (0.15,0){};
\end{tikzpicture}
\Big\},\;\;
\textbf{F}(3) = \Big\{\begin{tikzpicture}[baseline=-0.65ex]
\node[circle,fill,inner sep=1.5pt](t1) at (-0.3,0){};
\node[circle,fill,inner sep=1.5pt](t2) at (0,0){};
\node[circle,fill,inner sep=1.5pt](t3) at (0.3,0){};
\end{tikzpicture},\;
\begin{tikzpicture}[baseline=-0.65ex]
\node[circle,fill,inner sep=1.5pt](t1) at (-0.2,0.2){};
\node[circle,fill,inner sep=1.5pt](t2) at (0.2,0.2){};
\node[circle,fill,inner sep=1.5pt](t3) at (0,-0.2){};
\draw (t1)--(t3)--(t2);
\end{tikzpicture},\;
\begin{tikzpicture}[baseline=-0.65ex]
\node[circle,fill,inner sep=1.5pt](t1) at (0,0.2){};
\node[circle,fill,inner sep=1.5pt](t2) at (0,-0.2){};
\node[circle,fill,inner sep=1.5pt](t3) at (0.4,-0.2){};
\draw (t1)--(t2);
\end{tikzpicture},\;
\begin{tikzpicture}[baseline=-0.65ex]
\node[circle,fill,inner sep=1.5pt](t1) at (-0.3,0){};
\node[circle,fill,inner sep=1.5pt](t2) at (0,0){};
\node[circle,fill,inner sep=1.5pt](t3) at (0.3,0){};
\draw (t1)--(t2)--(t3);
\end{tikzpicture}
\Big\},\;\; \ldots
$$

As we can see, each electronic diagram $D_n$ is encoded by a rooted forest of weight $2n$. Precisely, each electronic diagram $D_n$ is encoded by \textbf{decorated forest}. By \textbf{decorated forest} we mean the forest where elements of a given {\bf alphabet} set $\mathbf{A}$ are attached to each vertices. For an electronic diagram $D_n$ we combine the set $\mathbf{A}_n=\{a_1,b_1,a_2,b_2,\ldots,a_n,b_n\}$.  \\

In applications to R-L ladder networks, the number $a_i$ will represent inertances and $b_i$ resistances.\\

Denoting by $S_n$ the basic element 
$$
\begin{tikzpicture}
\node[circle,fill,inner sep=2pt] (tree1) [node distance = 2cm]{};
		\node[] (texta) [below of = tree1,node distance = 0.45cm]{$L_n$};
		\node[circle,fill,inner sep=2pt] (tree2) [right of = tree1,node distance = 1cm]{};
		\node[] (textb) [below of = tree2,node distance = 0.45cm]{$R_n$};
\end{tikzpicture}  
$$
we can construct recursively the diagram $D_n$ using the notion of {\bf grafting map}: let $T_1 \dots T_k$ be a decorated forest and $a\in \mathbf{A}$. We denote by $\beta_a (T_1 \dots T_k )$ the tree having $a$ as a root, i.e.
\begin{equation}
\label{defgraft}
\beta_a (T_1 \dots T_k )=a (T_1 \otimes T_2 \dots \otimes T_k ). 
\end{equation}
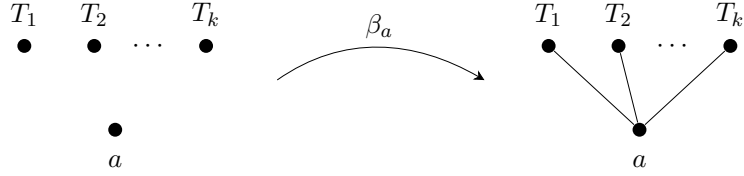
\begin{figure}[!ht]
    \resizebox{0.65\linewidth}{!}{	\begin{tikzpicture}
		[-,auto,thin,
			boxR/.style={rectangle,text width=4em,fill=blue!5,draw,align=center},
			boxC/.style={rectangle,text width=1em,fill=orange!5,draw,align=center},
			connection/.style={inner sep=0,outer sep=0},
			mid arrow/.style={draw, postaction={decorate},
				decoration={
					markings, mark=at position 0.75 with {\arrow[scale=2]{>}}}},
			end arrow/.style={draw, postaction={decorate},
				decoration={
					markings, mark=at position 1 with {\arrow[scale=2]{>}}}}	
			]

		
		\node[circle,fill,inner sep=2pt] (a) []{};
		\node[] (texta) [below of = a,node distance = 0.45cm]{$a$};
		
		\node[connection] (1) [above of = a, node distance = 1.2cm]{};
		
		\node[circle,fill,inner sep=2pt] (T1) [left of = 1,node distance = 1.3cm]{};
		\node[] (textT1) [above of = T1, node distance = 0.45cm]{$T_1$};
		
		\node[circle,fill,inner sep=2pt] (T2) [left of = 1,node distance = 0.3cm]{};
		\node[] (textT2) [above of = T2,node distance = 0.45cm]{$T_2$};
		
		\node[] (dots) [right of = T2, node distance = 0.8cm]{$\cdots$};
		
		\node[circle,fill,inner sep=2pt] (Tk) [right of = dots, node distance = 0.8cm]{};
		\node[] (textTk) [above of = Tk, node distance = 0.45cm]{$T_k$};
		
		
		\node[connection] (beta0) [below of = Tk,node distance = 0.5cm]{};
		\node[connection] (beta1) [right of = beta0, node distance = 1cm]{};
		\node[connection] (beta2) [right of = beta1, node distance = 3cm]{};		
		
		
		\node[circle,fill,inner sep=2pt] (a1) [right of = a,node distance = 7.5cm]{};
		\node[] () [below of = a1, node distance = 0.45cm]{$a$};
		
		\node[connection] (2) [above of = a1, node distance = 1.2cm]{};
		
		\node[circle,fill,inner sep=2pt] (T1-1) [left of = 2,node distance = 1.3cm]{};
		\node[] () [above of = T1-1, node distance = 0.45cm]{$T_1$};
		
		\node[circle,fill,inner sep=2pt] (T2-1) [left of = 2,node distance = 0.3cm]{};
		\node[] () [above of = T2-1,node distance = 0.45cm]{$T_2$};
		
		\node[] (dots) [right of = T2-1, node distance = 0.8cm]{$\cdots$};
		
		\node[circle,fill,inner sep=2pt] (Tk-1) [right of = dots, node distance = 0.8cm]{};
		\node[] () [above of = Tk-1, node distance = 0.45cm]{$T_k$};
		
		\draw (a1) -- (T1-1);
		\draw (a1) -- (T2-1);
		\draw (a1) -- (Tk-1);
		\draw [-{Stealth[length=1.5mm, width=1.5mm]}] (beta1) to [out=35,in=150] node[above] {$\beta_a$} (beta2);
					
	\end{tikzpicture}}
\caption{Grafting map} \label{grafting}
\end{figure}
Using this definition, we can of course graft a tree on a given forest by selecting a particular vertices. 

We then have 
\begin{equation}
D_1 =S_1 ,\ D_2 =\beta_{a_1} (S_2 D_1) ,\ \ D_3 =\beta_{a_2} (S_3 D_2 ),\dots 
\end{equation}

This can be seen as a composition of grafting maps:
\begin{equation}
\label{recurform}
D_n =\beta_{a_{n-1}} (S_n \beta_{a_{n-2}} (S_{n-1} \dots \beta_{a_1} (S_2 S_1 )\dots )) .
\end{equation}
\begin{figure}[!ht]
\centering
\begin{minipage}[b]{0.46\textwidth}
    \centering
    \resizebox{1\linewidth}{!}{	\begin{tikzpicture}
		[-,auto,thin,decoration={brace,raise=2pt},
			boxR/.style={rectangle,text width=4em,fill=blue!5,draw,align=center},
			boxC/.style={rectangle,text width=1em,fill=orange!5,draw,align=center},
			connection/.style={inner sep=0,outer sep=0},
			mid arrow/.style={draw, postaction={decorate},
				decoration={
					markings, mark=at position 0.75 with {\arrow[scale=2]{>}}}},
			end arrow/.style={draw, postaction={decorate},
				decoration={
					markings, mark=at position 1 with {\arrow[scale=2]{>}}}}	
			]

		
		
		\node[circle,fill,inner sep=2pt] (a1) []{};
		\node[] (ta1) [below of = a1,node distance = 0.45cm]{$a_1$};
		\node[circle,fill,inner sep=2pt] (b1) [right of = a1,node distance = 1cm]{};
		\node[] (tb1) [below of = b1,node distance = 0.4cm]{$b_1$};
				
		\node[connection] (T-1) [above of = a1,node distance = 1.2cm]{};
		\node[circle,fill,inner sep=2pt] (a2) [left of = T-1,node distance = 0.5cm]{};
		\node[] (ta2) [left of = a2,node distance = 0.45cm]{$a_2$};
		\node[circle,fill,inner sep=2pt] (b2) [right of = T-1,node distance = 0.5cm]{};
		\node[] (tb2) [right of = b2,node distance = 0.45cm]{$b_2$};
		
		\node[connection] (S1-1) [below left of = ta2,node distance=0.8cm] {};
		\node[connection] (S1-2) [above left of = ta2,node distance=0.8cm] {};
		\node[connection] (S1-3) [above right of = tb2,node distance=0.8cm] {};
		\node[connection] (S1-4) [below right of = tb2,node distance=0.8cm] {};
		
		\node[connection] (beta0) [above right of = b1,node distance = 0.7cm]{};
		\node[connection] (beta1) [right of = beta0, node distance = 0.5cm]{};
		\node[connection] (beta2) [right of = beta1, node distance = 2.5cm]{};
		
		\draw[densely dashed] (S1-1) -- (S1-2) -- node[above]{$S_2$} (S1-3) -- (S1-4) -- (S1-1);
		
		\draw[decorate,thick] (tb1.south east) -- (ta1.south west|-tb1.south) node[midway,below=4pt]{$D_1 = S_1$};
		
		\draw [-{Stealth[length=1.5mm, width=1.5mm]}] (beta1) to [out=35,in=150] node[above] {$\beta_{a_1}(S_2D_1)$} (beta2);

		\node[connection] (beta3) [right of = beta2, node distance = 1cm]{};
		\node[circle,fill,inner sep=2pt] (a1-1) [below right of = beta3,node distance = 0.7cm]{};
		\node[] (ta1-1) [below of = a1-1,node distance = 0.45cm]{$a_1$};
		\node[circle,fill,inner sep=2pt] (b1-1) [right of = a1-1,node distance = 1cm]{};
		\node[] (tb1-1) [below of = b1-1,node distance = 0.4cm]{$b_1$};
		
		\node[connection] (T-2) [above of = a1-1,node distance = 1.2cm]{};
		\node[circle,fill,inner sep=2pt] (a2-1) [left of = T-2,node distance = 0.5cm]{};
		\node[] (ta2-1) [left of = a2-1,node distance = 0.45cm]{$a_2$};
		\node[circle,fill,inner sep=2pt] (b2-1) [right of = T-2,node distance = 0.5cm]{};
		\node[] (tb2-1) [right of = b2-1,node distance = 0.45cm]{$b_2$};
		
		\draw[decorate,thick] (tb1-1.south east) -- (ta1-1.south west|-tb1-1.south) node[midway,below=4pt]{$D_2$};
		\draw (a1-1) -- (a2-1);
		\draw (a1-1) -- (b2-1);
					
	\end{tikzpicture}}
\end{minipage}\hspace{9mm}
\begin{minipage}[b]{0.46\textwidth}
    \centering
    \resizebox{1\linewidth}{!}{	\begin{tikzpicture}
		[-,auto,thin,decoration={brace,raise=2pt},
			boxR/.style={rectangle,text width=4em,fill=blue!5,draw,align=center},
			boxC/.style={rectangle,text width=1em,fill=orange!5,draw,align=center},
			connection/.style={inner sep=0,outer sep=0},
			mid arrow/.style={draw, postaction={decorate},
				decoration={
					markings, mark=at position 0.75 with {\arrow[scale=2]{>}}}},
			end arrow/.style={draw, postaction={decorate},
				decoration={
					markings, mark=at position 1 with {\arrow[scale=2]{>}}}}	
			]

		
		
		\node[circle,fill,inner sep=2pt] (a1) []{};
		\node[] (ta1) [below of = a1,node distance = 0.45cm]{$a_1$};
		\node[circle,fill,inner sep=2pt] (b1) [right of = a1,node distance = 1cm]{};
		\node[] (tb1) [below of = b1,node distance = 0.4cm]{$b_1$};
		
		\node[connection] (T-1) [above of = a1,node distance = 1.2cm]{};
		\node[circle,fill,inner sep=2pt] (a2) [left of = T-1,node distance = 0.5cm]{};
		\node[] (ta2) [left of = a2,node distance = 0.45cm]{$a_2$};
		\node[circle,fill,inner sep=2pt] (b2) [right of = T-1,node distance = 0.5cm]{};
		\node[] (tb2) [right of = b2,node distance = 0.45cm]{$b_2$};
		
		\draw[decorate,thick] (tb1.south east) -- (ta1.south west|-tb1.south) node[midway,below=4pt]{$D_2$};
		\draw (a1) -- (a2);
		\draw (a1) -- (b2);
		
		\node[connection] (T-2) [above of = a2,node distance = 1.2cm]{};
		\node[circle,fill,inner sep=2pt] (a3) [left of = T-2,node distance = 0.5cm]{};
		\node[] (ta3) [left of = a3,node distance = 0.45cm]{$a_3$};
		\node[circle,fill,inner sep=2pt] (b3) [right of = T-2,node distance = 0.5cm]{};
		\node[] (tb3) [right of = b3,node distance = 0.45cm]{$b_3$};
		
		\node[connection] (S1-1) [below left of = ta3,node distance=0.8cm] {};
		\node[connection] (S1-2) [above left of = ta3,node distance=0.8cm] {};
		\node[connection] (S1-3) [above right of = tb3,node distance=0.8cm] {};
		\node[connection] (S1-4) [below right of = tb3,node distance=0.8cm] {};
		
		\node[connection] (beta0) [above right of = b2,node distance = 0.5cm]{};
		\node[connection] (beta1) [right of = beta0, node distance = 0.7cm]{};
		\node[connection] (beta2) [right of = beta1, node distance = 2.5cm]{};
		
		\draw[densely dashed] (S1-1) -- (S1-2) -- node[above]{$S_3$} (S1-3) -- (S1-4) -- (S1-1);
		
		\draw [-{Stealth[length=1.5mm, width=1.5mm]}] (beta1) to [out=35,in=150] node[above] {$\beta_{a_2}(S_3D_2)$} (beta2);

		\node[circle,fill,inner sep=2pt] (a1-1) [right of = b1, node distance = 4.7cm]{};
		\node[] (ta1-1) [below of = a1-1,node distance = 0.45cm]{$a_1$};
		\node[circle,fill,inner sep=2pt] (b1-1) [right of = a1-1,node distance = 1cm]{};
		\node[] (tb1-1) [below of = b1-1,node distance = 0.4cm]{$b_1$};
		
		\node[connection] (T1-1) [above of = a1-1,node distance = 1.2cm]{};
		\node[circle,fill,inner sep=2pt] (a2-1) [left of = T1-1,node distance = 0.5cm]{};
		\node[] (ta2-1) [left of = a2-1,node distance = 0.45cm]{$a_2$};
		\node[circle,fill,inner sep=2pt] (b2-1) [right of = T1-1,node distance = 0.5cm]{};
		\node[] (tb2-1) [right of = b2-1,node distance = 0.45cm]{$b_2$};
		
		\draw[decorate,thick] (tb1-1.south east) -- (ta1-1.south west|-tb1-1.south) node[midway,below=4pt]{$D_3$};
		\draw (a1-1) -- (a2-1);
		\draw (a1-1) -- (b2-1);
		
		\node[connection] (T1-2) [above of = a2-1,node distance = 1.2cm]{};
		\node[circle,fill,inner sep=2pt] (a3-1) [left of = T1-2,node distance = 0.5cm]{};
		\node[] (ta3-1) [left of = a3-1,node distance = 0.45cm]{$a_3$};
		\node[circle,fill,inner sep=2pt] (b3-1) [right of = T1-2,node distance = 0.5cm]{};
		\node[] (tb3-1) [right of = b3-1,node distance = 0.45cm]{$b_3$};
		
		\draw (a2-1) -- (a3-1);
		\draw (a2-1) -- (b3-1);
					
	\end{tikzpicture}} 
\end{minipage}
\caption{Recursive construction of D2 and D3} \label{DDDD}
\end{figure}
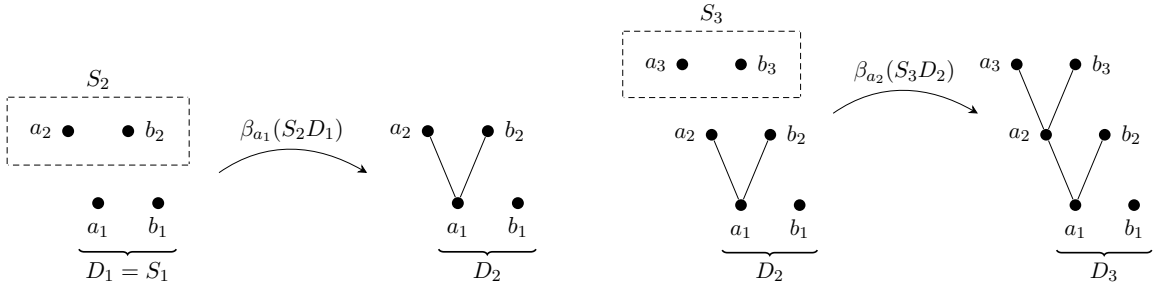

As the forest $S_n$ have a very particular structure, we can even give a more explicit form:
\begin{equation}
D_1 = a_1 \otimes b_1,\ D_2 =a_1 (a_2 \otimes b_2) \otimes b_1 ,\ D_3 =a_1 (a_2 (a_3 \otimes b_3)\otimes b_2 )\otimes b_1 ,\dots .
\end{equation}

These relations will be useful by deriving the recursive form of the admittance for the diagrams $D_n$, $n\geq 1$.

\section{Explicit computation of admittance/impedance for recursive diagram \texorpdfstring{$(D_n)_{n\geq 1}$}{Dn (n>=1)}}
\label{sectionrecursive}

\subsection{General properties of admittance for electronic diagrams}

Let us consider an electronic diagram encoded by a decorated forest $F$. We are interested in the explicit computation of the associated admittance denoted by $Y(F)$ and impedance denoted by $Z(F)$ satisfying 
\begin{equation}
Y(F)=\di\frac{1}{Z(F)} .
\end{equation}
The mapping $Y:\mathscr{F} \rightarrow \R$ can be computed recursively on the weight of the Forest. Indeed, we have the following well known algebraic relations for the computations of $Y$:

\begin{lemma}
Let $F_1$ and $F_2$ be two forest. Then, we have 
\begin{equation}
\label{algrules}
\left .
\begin{array}{lll}
Y(F_1 F_2 )=\di\frac{1}{Y(F_1)^{-1}+Y(F_2)^{-1}} ,\\
Y(F_1 \otimes F_2 ) = Y(F_1 )+Y(F_2 ) .
\end{array}
\right .
\end{equation}
\end{lemma}

Using this rules, the admittance of any decorated forest with decorations $\{ a_1,\dots ,a_k \}$ can be computed explicitly using the single admittance $Y(a_i )$. Moreover, as the algebraic rules are only using addition and inversion, we have the following structural lemma:

\begin{lemma}
\label{struclemma}
For any decorated forest of finite weight $n$ with decorations $\{ a_1 ,\dots ,a_n\}$, the corresponding admittance is a rational function of the single admittance $Y(a_1),\dots ,Y(a_n )$. 
\end{lemma}

This lemma has an interesting consequence for $R-L$ ladder network as in this case the single admittance have for all $\omega\in \R$ (resp. $\R^*$) a very special form given by: 
\begin{equation}
Y(R) (\omega )=\di\frac{1}{R},\ \ \mbox{\rm and}\ \ \ Y(L)(\omega ) = \di\frac{1}{jL\omega},\ \ \mbox{\rm with}\ j^2 =-1 .
\end{equation}
Then, using the structural lemma \ref{struclemma}, we have :

\begin{lemma}
\label{strucRL}
Let $(D_n )_{n\geq 1}$ be a recursive diagram with decorations made of resistances $R_i$ and inertance $L_i$. Then, for all $n \geq 1$, the admittance $Y(D_n) (\omega )$ is a rational function of $\omega$. 
\end{lemma}

The explicit form of the asymptotic impedance is nevertheless not so easy to catch and we need to have more information on the structure and properties of this function when $n$ goes to infinity.

\subsection{Admittance of recursive diagrams}

Let us consider a recursive diagram $(D_n )_{n\geq 1}$ with decorations $\{ a_1,b_1 ,a_2,b_2 ,\dots ,a_n,b_n,\dots \}$. Formula \eqref{recurform} can be written as 
\begin{equation}
D_2 =\beta_{a_1} (S_2 S_1 )  =a_1 S_2 \otimes b_1 ,\ 
D_3 = a_1 (\beta_{a_2} (S_3 S_2 )) \otimes b_1 =a_1 (a_2 S_3 \otimes b_2 ) \otimes b_1 ,\dots 
\end{equation}
which can be written for all $n\geq 1$ as 
\begin{equation}
\label{recurform2}
D_{n+1} =a_1 (a_2 (\dots a_{n-1}(a_n S_{n+1} \otimes b_n)\otimes b_{n-1} \dots)\otimes b_2) \otimes b_1  .
\end{equation}
A more convenient way to write this relation is to denote by $F(a_1,b_1,\dots ,a_n,b_n)$ the decorated forest associated to $D_n$. Equation \eqref{recurform2} is then given by 
\begin{equation}
\label{forestrecursiv}
F(a_1,b_1,\dots ,a_n,b_n )= a_1 F(a_2,b_2 ,\dots ,a_n,b_n ) \otimes b_1 .
\end{equation}

A direct consequence of the algebraic rules \eqref{algrules} is that the admittance of $F(a_1,b_1,\dots ,a_n,b_n)$ satisfies
\begin{equation}
\label{recuradmi}
Y(F(a_1,b_1,\dots ,a_n,b_n))=Y(b_1 ) +\di\frac{1}{Y(a_1)^{-1} +Y(F(a_2 ,b_2,\dots ,a_n,b_n ))^{-1}}
,
\end{equation}
for all $n\geq 1$.\\ 

Using the classical notation for {\bf continued fractions} given by 
\begin{equation}
[\alpha_1 ,\dots ,\alpha_n ] =\alpha_1 +\di\frac{1}{\alpha_2 +\di\frac{1}{\ddots +\di\frac{1}{\alpha_n}}} , 
\end{equation}
we can rewrite relation \eqref{recuradmi} as
\begin{equation}
\label{recuradmi2}
Y(F(a_1,b_1,\dots ,a_n,b_n))=[Y(b_1 ),Y(a_1)^{-1} ,Y(F(a_2 ,b_2,\dots ,a_n,b_n ))]
.
\end{equation}
for all $n\geq 1$.\\ 

Applying recursively relation \eqref{recuradmi2}, we obtain the following explicit form of the admittance:

\begin{lemma}
For any recursive diagram $(D_n )_{n\geq 1}$ with decorations $\{ a_1,b_1 ,a_2,b_2 ,\dots ,a_n,b_n,\dots \}$ and $F(a_1,b_1,\dots ,a_n,b_n)$ as the decorated forest associated to $D_n$, $n\geq 1$, we have 
\begin{equation}
\label{recursiveformula}
Y(F(a_1,b_1,\dots ,a_n,b_n))=[Y(b_1),Y(a_1)^{-1},Y(b_2),Y(a_2)^{-1},\dots ,Y(b_n),Y(a_n)^{-1}] .    
\end{equation}
\end{lemma}

In order to go further, we need some information on the decorations. This is done in the next Section by considering fractal R-L networks.

\section{Admittance for fractal R-L network}
\label{sectioncomputation}
\subsection{Fractal networks}

We begin with a definition of a {\bf fractal diagram}:

\begin{definition}
A recursive diagram $(D_n)_{n\geq 1}$ with decorations $\{ a_1,b_1 ,a_2,b_2 ,\dots ,a_n,b_n,\dots \}$ is called fractal if for all $i\geq 1$, we have 
\begin{equation}
\label{fracrel}
a_{i+1} = \sigma a_i,\ \ \ b_{i+1} = \rho b_i ,
\end{equation}
for some real numbers $\sigma,\rho \in \R^*$. Such a fractal recursive diagram will be denoted by $(\mathscr{F}_n^{\sigma ,\rho} (a_1 ,b_1 )_{n\geq 1}$ where $\mathscr{F}_n^{\sigma ,\rho}$ corresponds to $D_n$ with the corresponding decoration satisfying \eqref{fracrel}. 
\end{definition}

We can simplify our notations for fractal diagrams. Let us denote for all $n\geq 1$ by $F_n^{\sigma ,\rho} (a_1 ,b_1)$ the decorated forest 
\begin{equation}
F_n^{\sigma ,\rho} (a_1,b_1)=F(a_1,b_1, \sigma a_1 ,\rho b_1 ,\dots ,\sigma ^{n-1} a_1 ,\rho^{n-1} b_1 ) .    
\end{equation}

We then have the following Lemma:

\begin{lemma}
\label{asymadmittance}
Let $(\mathscr{F}_n^{\sigma ,\rho} )_{n\geq 1}$ be a fractal diagram. For all $n\geq 1$, we have 
\begin{equation}
\label{recurfractal}
Y(F_n^{\sigma ,\rho} (a_1 ,b_1) )=[Y(b_1 ),Y(a_1)^{-1},Y(F_{n-1}^{\sigma,\rho} (\sigma a_1 ,\rho b_1 )) ] .
\end{equation}
\end{lemma}

We now denote by $Y^{\sigma,\rho} _{a_1,b_1}$ the admittance of the asymptotic fractal diagram. Equation \eqref{recurfractal} then induces the following relation:

\begin{equation}
\label{asympfractal}
Y^{\sigma,\rho}_{a_1,b_1} =[Y(b_1 ),Y(a_1)^{-1},Y^{\sigma ,\rho}_{\sigma a_1,\rho b_1} ] .
\end{equation}

This relation is difficult to handle but in many application of fractal diagram a more stringent condition is made relating the two intrinsic parameters $\sigma$ and $\rho$ of the diagram. We detail this situation in the next Section.

\subsection{Admittance of Oustaloup and R-L fractal networks}

We first introduce the definition of Oustaloup fractal diagrams:

\begin{definition}
A fractal diagram $(\mathscr{F}_n^{\sigma ,\rho} (a_1,b_1))_{n\geq 1}$ is called a Oustaloup fractal diagram if the following relation between the fractal parameters holds
\begin{equation}
\rho =\sigma^{-1} .
\end{equation}
\end{definition}

Of special importance are Oustaloup fractal diagram where $a_1$ and $b_1$ are respectively inertance and resistance, then the Oustaloup fractal diagram corresponds to the fractal R-L ladder network. Denoting by $Y^{\sigma}_{a_1,b_1 }$ the function 
\begin{equation}
Y^{\sigma}_{a_1,b_1} = Y^{\sigma ,\sigma^{-1}}_{a_1 ,b_1} ,    
\end{equation}
we obtain the following Theorem:

\begin{theorem}
The asymptotic impedance $Y^{\sigma}_{a_1 ,b_1}$ of a fractal R-L ladder network with parameters $\sigma$, $a_1$ as inertance, $b_1$ as resistance, satisfies
\begin{equation}
\label{functionalrelation}
Y^{\sigma}_{a_1,b_1} (\omega ) =[Y(b_1),Y(a_1)^{-1} , \sigma Y^{\sigma}_{a_1,b_1} (\sigma^2 \omega ) ].
\end{equation}
\end{theorem}

The proof is a direct consequence of Lemma \ref{asymadmittance} and the following Lemma:

\begin{lemma}
\label{lemmeasympfrac}
The asymptotic impedance $Y^{\sigma}_{a_1 ,b_1}$ of a fractal R-L ladder network with parameters $\sigma$, $a_1$ as inertance, $b_1$ as resistance, satisfies
\begin{equation}
Y^{\sigma}_{\sigma a_1,\sigma^{-1} b_1} (\omega ) =\sigma Y^{\sigma}_{a_1,b_1} (\sigma^2 \omega ) .
\end{equation}
\end{lemma}

The proof is given in Section \ref{prooflemma}.

\section{Fractional behavior of infinite fractal R-L ladder networks}
\label{sectionfractionalbehavior}

The functional equation \eqref{functionalrelation} is difficult to solve. Following classical approaches, we make the two following hypothesis:

\begin{equation}
\label{asymptoticassumptions}
\left .
\begin{array}{l}
\sigma Y(a_1)^{-1} (\omega ) Y^{\sigma}_{a_1 ,b_1} (\sigma^2 \omega ) 
\underset{\omega \rightarrow 0}{\longrightarrow}
 0 ,\\
Y^{\sigma}_{a_1,b_1} (\sigma^2 \omega ) 
\underset{\omega \rightarrow 0}{\longrightarrow}
 +\infty .
\end{array}
\right .
\end{equation}

Under these assumptions, the asymptotic behavior of $Y^{\sigma}_{a_1 ,b_1} (\omega )$ is related to the simplified functional relation: 

\begin{equation}
\label{scaleinvariance}
Y (\omega  ) =\sigma 
Y (\sigma^2 \omega ) ,
\end{equation}
meaning that $Y$ is a scale invariant function.\\

We then look if fractional or power law behaviors are possible for functions satisfying the scale invariance relation \eqref{scaleinvariance}. Precisely, can we find solutions of \eqref{scaleinvariance} in the class
\begin{equation}
\label{fractionalbehavior}
Y (\omega ) =K \omega^{\gamma} ,
\end{equation}
for some real constants $K$ and $\gamma$ ? \\

Replacing directly $Y$ by a function \eqref{fractionalbehavior} in \eqref{scaleinvariance}, we obtain 
\begin{equation}
K \omega ^{\gamma} =\sigma K \sigma^{2\gamma} \omega^{\gamma}  ,
\end{equation}
so that $\gamma$ have to satisfy
\begin{equation}
\sigma^{2\gamma +1} =1 ,\ \forall\ \omega\not= 0 .
\end{equation}

We then deduce the following Theorem:

\begin{theorem}
Fractional behaviors of the form \eqref{fractionalbehavior} satisfying the scale invariance functional \eqref{scaleinvariance} for fractal R-L ladder networks exist if and only if $\gamma =-1/2$.
\end{theorem}

As a consequence, we are waiting for asymptotic behavior of the form 
\begin{equation}
Y^{\sigma}_{a_1,b_1} (\omega ) \underset{\omega \rightarrow 0}{\longrightarrow} \di\frac{K}{\sqrt{\omega}} .
\end{equation}

It must be noted that this result is coherent with our asymptotic assumptions 
\eqref{asymptoticassumptions}. \\

This kind of behavior suggest a possible connection between fractal R-L ladder networks and the diffusion equation. This fact is well known in the engineering community and we give details in the next Section.

\section{Fractal R-L ladder networks and the diffusion equation}
\label{sectiondiffusion}

The aim of this Section is to give a precise presentation of computations found for example in Oustaloup \cite{oustaloup}, \cite{riu} or more recently in \cite{sabatier} and \cite{clara}. These computations are not easy to handle and we hope that this section will allow more people to deal with.

\subsection{The continuous representation problem for R-L Ladder networks}

Let us consider a R-L ladder network described by the family of inertance and resistance $(L_n ,R_n)_{n\geq 1}$. For each unit of this ladder we have relations between the quantities $U_n$ and $I_n$ corresponding to the voltage and current: for each $n\geq 1$, we have 
\begin{equation}
\label{relationscircuit}
\left .
\begin{array}{lll}
U_n -U_{n-1} & = & -jL_{n-1} \omega I_n ,\\
I_{n+1} -I_n & = & -\di\frac{1}{R_n} U_n .
\end{array}
\right .
\end{equation}
\begin{center}
\begin{figure}[!ht]
    \resizebox{1.5\linewidth}{!}{	\begin{tikzpicture}
		[-,auto,thin,
			boxR/.style={rectangle,text width=3em,fill=blue!5,draw,align=center},
			boxL/.style={rectangle,text width=3em,fill=orange!5,draw,align=center},
			boxS/.style={rectangle,text width=3em,draw,align=center},
			connection/.style={inner sep=0,outer sep=0},
			mid arrow/.style={draw, postaction={decorate},
				decoration={
					markings, mark=at position 0.6 with {\arrow[scale=2]{>}}}},
			begin arrow/.style={draw, postaction={decorate},
				decoration={
					markings, mark=at position 0.1 with {\arrow[scale=2]{>}}}},
			begin arrow1/.style={draw, postaction={decorate},
				decoration={
					markings, mark=at position 0.3 with {\arrow[scale=2]{>}}}},
			end arrow/.style={draw, postaction={decorate},
				decoration={
					markings, mark=at position 1 with {\arrow[scale=2]{>}}}}	
			]

		\node[] (empty0) {};
		\node[connection] (I0) [right of = empty0,node distance = 1.2cm] {};
		\node[connection] (1I0) [above of = I0, node distance = 1.8cm] {};
		\node[connection] (2I0) [below of = I0, node distance = 1.8cm]{};
		
		\node[boxL] (L0) [right of = 1I0,node distance = 1.8cm] {$L_0$};
		\node[boxR] (R0) [right of = 2I0, node distance = 1.8cm]{$R_0$};
		
		\node[connection] (I1) [right of = L0,node distance=1.7cm]{};
		\node[connection] (dots) [right of = I1,node distance = 0.5cm]{$\cdots$};
		\node[connection] (I1-In) [above right of = dots,node distance=1cm]{};
		\node[connection] (In) [right of = I1-In,node distance=1.2cm]{};
		
		\node[connection] (1In) [above of = In, node distance = 1.3cm]{};
		\node[connection] (2In) [below of = In, node distance = 1.3cm]{};
		
		\node[boxL] (Ln) [right of = 1In,node distance = 1.8cm] {$L_n$};
		\node[boxR] (Rn) [right of = 2In,node distance = 1.8cm] {$R_n$};
		
		\node[connection] (In+1) [right of = Ln, node distance = 2cm]{};		
		\node[connection] (1In+1) [above of = In+1, node distance = 0.5cm]{};
		\node[connection] (2In+1) [below of = In+1, node distance = 0.5cm]{};
		\node[connection] (3In+1) [right of = 1In+1, node distance = 0.3cm]{};
		\node[connection] (4In+1) [right of = 2In+1, node distance = 0.3cm]{};
		\node[connection] (dots1) [right of = In+1,node distance = 0.8cm]{$\cdots$};
		
		\node[connection] (5In+1) [right of = 3In+1, node distance = 1cm]{};
		\node[connection] (6In+1) [right of = 4In+1, node distance = 1cm]{};
		\node[connection] (7In+1) [right of = 5In+1, node distance = 0.3cm]{};
		\node[connection] (8In+1) [right of = 6In+1, node distance = 0.3cm]{};
		
		\node[connection] (9In+1) [below of = 7In+1, node distance = 0.5cm]{};
		\node[connection] (10In+1) [right of = 9In+1, node distance = 0.5cm]{};
		
		\node[connection] (0Rn) [below of = 10In+1, node distance = 1.3cm]{};
		\node[connection] (1Rn) [below of = 10In+1, node distance = 2.6cm]{};
		\node[connection] (2Rn) [right of = 0Rn, node distance = 1cm]{};
		
		\node[connection] (dots2) [below right of = 2Rn,node distance = 1cm]{$\cdots$};
		\node[connection] (0I1) [right of = dots2,node distance = 0.5cm]{};
		\node[connection] (1I1) [right of = 0I1,node distance = 0.8cm]{};
		\node[connection] (2I1) [below of = 1I1,node distance = 3.6cm]{};
		\node[connection] (3I1) [below of = 1I1,node distance = 1.8cm]{};
		\node[connection] (4I1) [right of = 3I1,node distance = 1cm]{};
		
		
		\node[connection] (S1) [above of = I1-In,node distance = 2.5cm] {};
		\node[connection] (S2) [below of = I1-In,node distance = 3cm] {};
		\node[connection] (S3) [right of = S1,node distance = 8.1cm] {};
		\node[connection] (S4) [right of = S2,node distance = 8.1cm] {};

		\node[connection] (U01) [below of = 2I1,node distance = 0.8cm] {};
		\node[connection] (U02) [below of = 2I0,node distance = 0.8cm] {};
		
		\node[connection] (U11) [below left of = I1,node distance = 1.3cm] {};
		\node[connection] (U12) [left of = U11,node distance = 1.5cm] {};
		
		\node[connection] (Un1) [below of = 1Rn,node distance = 0.8cm] {};
		\node[connection] (Un2) [below of = 2In,node distance = 0.8cm] {};
		
		\node[connection] (1Un+1) [below of = 2In+1,node distance = 0.3cm] {};
		\node[connection] (2Un+1) [below of = 8In+1,node distance = 0.3cm] {};
		\node[connection] (3Un+1) [left of = 1Un+1,node distance = 1.3cm] {};
		\node[connection] (4Un+1) [left of = 3Un+1,node distance = 1.4cm] {};

					
		\draw[mid arrow,>=stealth] (empty0) -- (I0) node[above=1mm,text width=4em] {$I_{0}$};
		\draw (1I0) -- (2I0) {};		
		\draw[mid arrow,>=stealth] (1I0) -- (L0) {};
		\draw[mid arrow,>=stealth] (L0) -- (I1) node[above=1mm,text width=3em] {$I_1$};
		\draw[mid arrow,>=stealth] (I1-In) -- (In) node[above=1mm,text width=4em] {$I_n$};
			
		\draw (1In) -- (2In) {};	
		\draw[mid arrow,>=stealth] (1In) -- (Ln) {};	
		\draw[mid arrow,>=stealth] (2In) -- (Rn) {};		
		\draw[mid arrow,>=stealth] (Ln) -- (In+1) node[above=1mm,text width=5em] {$I_{n+1}$};
		\draw (4In+1)--(2In+1)--(1In+1)--(3In+1);
		\draw (5In+1)--(7In+1)--(8In+1)--(6In+1);
		\draw (9In+1)--(10In+1)--(1Rn);
		\draw[begin arrow1,>=stealth] (Rn)--(1Rn) node[above=1mm,text width=15em] {$I_{n}-I_{n+1}$};
		\draw[mid arrow,>=stealth] (0Rn) -- (2Rn);
		\draw (0I1)--(1I1)--(2I1);
		\draw[begin arrow,>=stealth] (R0)--(2I1) node[above=1mm,text width=61em] {$I_{0}-I_1$};
		\draw[mid arrow,>=stealth] (2I0)--(R0);
		\draw[mid arrow,>=stealth] (3I1) -- (4I1);

		\draw[densely dashed,draw=blue] (S1) -- (S3) -- (S4) -- (S2) -- (S1);
		
		\draw[end arrow,>=stealth] (U01) -- (U02) node[below=1mm,text width = -38em]{$U_0$};
		
		\draw[end arrow,>=stealth] (U11) -- (U12) node[below=1mm,pos=0.4]{$U_0-U_{1}$};
		
		\draw[end arrow,>=stealth] (Un1) -- (Un2) node[below=1mm,text width = -15em]{$U_n$};
		
		\draw[end arrow,>=stealth] (2Un+1) -- (1Un+1) node[below=1mm,text width = -2em]{$U_{n+1}$};
		
		\draw[end arrow,>=stealth] (3Un+1) -- (4Un+1) node[below=1mm,pos=0.3]{$U_n-U_{n+1}$};
					
	\end{tikzpicture}}
\caption{Structure of R-L ladder networks} \label{str}
\end{figure}
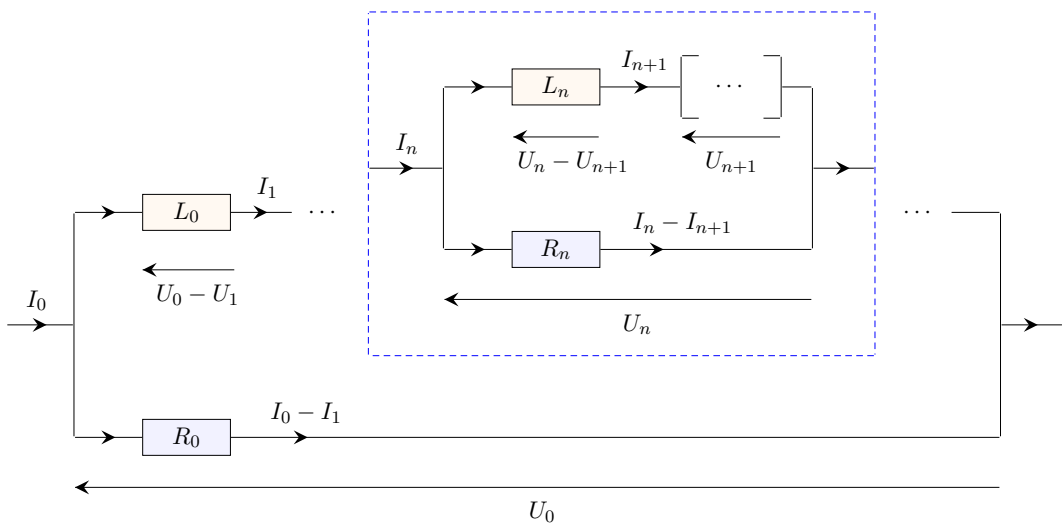
\end{center}

The main idea of the representation of R-L ladder networks by a partial differential equation (PDE), is to understand equations \eqref{relationscircuit} as coming from the discretization in the space variable $z$ of functions $U(z,\omega )$, $I(z,\omega )$, $R(z)$ and $L(z)$ over a particular discrete space-scale $\T =\{ z_n \}_{n\geq 1}$ to be determined. \\

The form of the discrete equations will be written as differences equations on the variable $z$ and this will induce a particular partial differential equation.\\

Precisely, we interpret \eqref{relationscircuit} as
\begin{equation}
\label{assumptioncontinuous}
U(z_n ,\omega )=U_n (\omega ), I(z_n ,\omega )=I_n (\omega ),\ R (z_n ) =R_n,\ L (z_n) =L_n . 
\end{equation}
Then, for each $n\geq 1$, we can rewrite \eqref{relationscircuit} as 
\begin{equation}
\label{systemini}
\left .
\begin{array}{lll}
\Delta_- [U] (z_n ,\omega ) & = & 
-j\left ( 
\di\frac{L(z_{n-1})}{z_n -z_{n-1}} 
\right ) 
\omega I(z_n ,\omega ),\\
\Delta_+ [I](z_n ,\omega ) & = & - 
\left ( 
\di\frac{1}{R(z_n ) (z_{n+1} -z_n)}
\right ) 
U(z_n ,\omega ) ,
\end{array}
\right .
\end{equation}
where $\Delta_+$, $\Delta_-$ are the classical forward and backward finite difference operators defined for all $f\in C (\T ,\R^d )$ by 
\begin{equation}
\Delta_+ [f](z_k)=\di\frac{f(z_{k+1})-f(z_k)}{z_{k+1}-z_k}\ \ \mbox{\rm and}\ \ 
\Delta_- [f](z_k)=\di\frac{f(z_k )-f(z_{k-1})}{z_k-z_{k-1}},
\end{equation}
respectively.\\

In order to recover the underlying PDE, we can write the second order equation
\begin{equation}
\label{finitePDE}
\Delta_- [-R(z) (\mu_{+}(z) -z) \Delta_+ [I]] (z_n ,\omega )  =  
-j\left ( 
\di\frac{L(\mu_{-} (z_n))}{z_n -\mu_{-} (z_n)}  
\right ) 
\omega I(z_n ,\omega )
\end{equation}
where $\mu_{+}(z_k) = z_{k+1}$, $\mu_{-}(z_k) = z_{k-1}$.

It is not easy to determine the continuous form of the finite difference equation \eqref{finitePDE}. \\

An idea is to simplify as far as possible the previous expression in such a way that the continuous PDE will be easy to identify. 

\subsection{Representation by a classical diffusion equation}

The easiest situation is obtained by imposing the following conditions called the {\bf diffusion conditions} as the associated continuous PDE will corresponds to the Fourier transform of the classical diffusion equation:\\

{\bf Diffusion conditions}: For all $z\in \T$, we have 
\begin{equation}
\label{diffusioncond}
R(z) (\mu_+ (z) -z) =R_0\ \ \mbox{\rm and}\ \ \ 
\di\frac{L(\mu_-(z))}{z -\mu_- (z)} = L_0 .
\end{equation}

If such a discrete space-scale $\T$ and functions $R$ and $L$ exist, then equation \eqref{finitePDE} reduces to 
\begin{equation}
\label{finitePDE2}
\Delta_- [R_0 \Delta_+ [I]] (z ,\omega )  =  j L_0 \omega I(z ,\omega ) ,
\end{equation}
for all $z\in \T^{\pm}$.\\

A continuous analogue is then given by 
\begin{equation}
\di\frac{\partial^2 I}{\partial z^2} = j\di\frac{L_0}{R_0} \omega I .
\end{equation}

However, one has to prove that a space-scale $\T$ and functions $R$ and $L$ solutions of the diffusion conditions can be indeed constructed. 

\subsection{Solution to the diffusion representation problem: the geometric case}

The determination of $R$ and $L$ from the conditions \eqref{diffusioncond} are in general difficult to handle. However, there exists a simple non trivial case called the {\bf geometric case}:\\

We assume that the space-scale $\T$ is such that for all $z\in \T^-$
\begin{equation}
\mu_+ (z) -z =\delta z ,
\end{equation}
where $\delta$ is a constant.\\ 

In order for $\T$ to be non trivial we must assume that $z_1 >0$ (or $<0$). Then, we obtain $z_2 =(1+\delta ) z_1$, $z_3 =(1+\delta )^2 z_1$, $\dots$. This corresponds to a {\bf geometric} distribution of the elements of $\T$.\\

Assuming that $\T$ is a geometric discrete space-scale, the functions $R$ and $L$ are then given by 
\begin{equation}
R(z)=\di\frac{R_0}{\delta z},\ \ \ L(z)=L_0 \delta z ,
\end{equation}
for all $z\in \T$.

As a consequence, we obtain the following Lemma:

\begin{lemma}[Diffusion conditions-geometric case]
\label{RLgeo}
Let $\T$ be a geometric space-scale with parameter $\delta \not=0$. A solution to the diffusion conditions \eqref{diffusioncond} is given by the functions $R$ and $L$ defined on $\R^*$ and $\R$ respectively defined by 
\begin{equation}
R(z)=\di\frac{R_0}{\delta z},\ \ \ L(z)=L_0 \delta z .
\end{equation}
\end{lemma}

It must be noted that other possibilities are certainly possible. However, a global characterization of the space-scale $\T$ leading to an easy identification of the functions $R$ and $L$ satisfying the diffusion conditions is out of the scope of this article. 

\subsection{Diffusion conditions, geometric space-scale and fractal R-L ladder networks}

The previous results impose some constraints on the type of R-L ladder networks that one can represent by a diffusion equation. Indeed, taking the functions $R$ and $L$ as in Lemma \ref{RLgeo} and reminding that their evaluation on the geometric space-scale $\T$ with parameter $\delta$ gives the coefficients $(R_n ,L_n )_{n\geq 1}$ of the R-L ladder network, we deduce :

\begin{lemma}
\label{RLfuncfrac}
A recursive R-L ladder network $(L_n, R_n )_{n\geq 1}$ corresponding to the discretization of the functions $R(z)=\di\frac{R_0}{\delta z}$ and $L(z)=L_0 \delta z$ over a geometric space-scale $\T$ with parameter $\delta\not=0$ is a fractal R-L ladder network with parameter $1+\delta$.
\end{lemma}

\begin{proof}
By assumption, we have 
\begin{equation}
\label{fractalladdercondition}
L (z_{n+1} ) =L_0 \delta  z_{n+1}=L_0 \delta (1+\delta) z_n =(1+\delta) L(z_n) ,
\end{equation}
and
\begin{equation}
R(z_{n+1} ) = \di\frac{R_0}{\delta z_{n+1}} =\di\frac{R_0}{\delta (1+\delta ) z_n } =(1+\delta )^{-1} R(z_n ) .
\end{equation}
As a consequence, the coefficients $(L_n ,R_n )_{n\geq 1}$ satisfy the relations
\begin{equation}
L_{n+1} =(1+\delta ) L_n,\ \ \ R_{n+1} =(1+\delta )^{-1} R_n ,
\end{equation}
which coincide with the conditions \eqref{fractalladdercondition} with parameter $1+\delta$. Then the family $(L_n ,R_n )_{n\geq 1}$ corresponds to a fractal R-L ladder network with parameter $1+\delta$. 
\end{proof}

We can now formulate the main result of this section:

\begin{theorem}
\label{main}
The behavior of a given current $I_0$ through a fractal R-L ladder network with inertance $L_0$, resistance $R_0$ and a scaling factor $1+\delta$ can be obtained by considering the discretization of the Fourier transform of the diffusion equation
\begin{equation}
\di\frac{\partial^2 \mathscr{I}}{\partial z^2} =\di\frac{L_0}{2\pi R_0} \di\frac{\partial \mathscr{I}}{\partial t} ,
\end{equation}
over a geometric space scale $\T$ with parameter $\delta$ given by 
\begin{equation}
\Delta_- [R_0 \Delta_+ [I]] (z ,\omega )  =  j L_0 \omega I(z ,\omega ) .
\end{equation}
\end{theorem}

This result explains in particular the asymptotic fractional behavior of fractal R-L ladder networks. 

\begin{proof}
Using Lemma \ref{RLfuncfrac}, the fractal R-L network with inertance $L_0$ and resistance $R_0$ with scaling factor $1+\delta$ can be recovered by discretization over the geometric space scale $\T$ of functions $R(z)=\di\frac{R_0}{\delta z}$ and $L(z)=L_0 \delta z$. These functions satisfy the diffusion conditions so that the behavior of a current $I_0$ though the corresponding fractal R-L ladder network can be seen as the discretization of the equation 
\begin{equation}
\di\frac{\partial^2 I}{\partial z^2} = j\di\frac{L_0}{R_0} \omega I ,
\end{equation}
for $z\in \T$ and $I_n =I(z_n)$, $z_n \in \T$. \\

Denoting by $\mathscr{F}[g]$ the Fourier transform of a function $g$ defined by 
\begin{equation}
\mathscr{F}[g] (\omega)=\di\int_{\R} \di e^{-2j\pi \omega s} g(s) ds ,
\end{equation}
and by $\mathscr{I}(z,t)$ the function defined by 
\begin{equation}
\mathscr{F} [\mathscr{I}(z,\cdot)](\omega )=I(z,\omega)
\end{equation}
and using the fact that
\begin{equation}
\mathscr{F} (f') (\omega )=2\pi j \omega \mathscr{F} (f) ,
\end{equation}
we obtain that $\mathscr{I}$ satisfies a diffusion equation with constant coefficient, precisely
\begin{equation}
\left .
\begin{array}{lll}
\di\frac{\partial^2 \mathscr{F}(\mathscr{I} (z,\cdot ))(\omega)}{\partial z^2} & = &  \di\frac{L_0}{R_0} j\omega \mathscr{F}(\mathscr{I} (z,\cdot ))(\omega) ,\\
\di\mathscr{F} \left ( \frac{\partial^2 \mathscr{I} (z,\cdot )}{\partial z^2} \right ) (\omega ) & = &  \mathscr{F}\left ( \di\frac{L_0}{2\pi R_0} \di\frac{\partial \mathscr{I} (z,\cdot )}{\partial t} (z,\cdot ) \right ) (\omega) ,
\end{array}
\right .
\end{equation}
which leads to
\begin{equation}
\di\frac{\partial^2 \mathscr{I}}{\partial z^2} =\di\frac{L_0}{2\pi R_0} \di\frac{\partial \mathscr{I}}{\partial t} .
\end{equation}
This concludes the proof.
\end{proof}

\section{Conclusion and perspectives}
\label{conclusion}

This article was designed to provide a mathematical introduction to classical work in engineering about the modeling of fractional behavior using special electronic devices and in particular fractal constructions called fractal R-L ladder networks. If these results were of importance in the engineering community it was related to the fact that such representations are able to provide modeling of a given phenomenon with less parameters than a classical approach meaning that such procedure give an explicit reduction of parameters for models (see for applications and results the work of D. Riu in \cite{riu}). Moreover, we wanted to introduce the connection between these fractal R-L ladder networks and the diffusion equation which was in fact at the beginning of this point of view. \\

The previous result suggest also, as already claimed by J. Sabatier and co-workers in \cite{sabatier} that the modeling of fractional behaviors can be done perhaps more efficiently looking for 
some general form of the diffusion equation, in particular with non-constant diffusion coefficients. This is done for example by J. Sabatier and al. in \cite{sabatier}. A mathematical treatment of these results with be given in a forthcoming article.

\section{Proof of Lemma \ref{lemmeasympfrac}}
\label{prooflemma}

We begin with a general property of continued fractions:

\begin{lemma}
\label{continuedfraction}
For all $\alpha \not=0$, we have 
\begin{equation}
[\alpha a_1,\alpha^{-1} a_2 ,\alpha a_3 ,\alpha^{-1}a_4 ,\dots ] =\alpha [a_1,a_2,\dots ] .
\end{equation}
\end{lemma}

\begin{proof}
The proof is done by induction. We have 
\begin{equation}
\label{resn1}
[\alpha a_1, \alpha^{-1} a_2]= \alpha a_1 +\di\frac{1}{\alpha^{-1} a_2} =\alpha \left ( 
a_1 +\di\frac{1}{a_2} \right ) =\alpha [a_1 ,a_2] .
\end{equation}
Assume that 
\begin{equation}
\label{induc}
[\alpha a_1 ,\alpha^{-1}a_2 ,\dots ,\alpha a_{2k+1}, \alpha^{-1}a_{2k+2}] =
\alpha [a_1 ,\dots ,a_{2k+2}] ,
\end{equation}
if true for all $k=0,\dots ,n$. \\

As we have 
\begin{equation}
[a_1 ,a_2,\dots ]=[a_1,[a_2,\dots ]],
\end{equation}
we deduce that for $k=n+1$ we have 
\begin{equation}
\left .
\begin{array}{lll}
[\alpha a_1 ,\alpha^{-1}a_2 ,\dots ,\alpha a_{2n+3}, \alpha^{-1}a_{2n+4}] & = & 
[\alpha a_1 ,[\alpha^{-1}a_2 ,\dots ,\alpha a_{2n+3}, \alpha^{-1}a_{2n+4}]] ,\\
& = & [\alpha a_1 ,[\alpha^{-1}a_2 ,[\alpha a_3,\dots ,\alpha a_{2n+3}, \alpha^{-1}a_{2n+4}]]] .
\end{array}
\right .
\end{equation}
Using the induction hypothesis \eqref{induc} and then equality \eqref{resn1}, we obtain 
\begin{equation}
\left .
\begin{array}{lll}
[\alpha a_1 ,\alpha^{-1}a_2 ,\dots ,\alpha a_{2n+3}, \alpha^{-1}a_{2n+4}] & = & 
[\alpha a_1 ,[\alpha^{-1}a_2 ,\alpha [a_3 ,\dots , a_{2n+4}]]] ,\\
 & = & [\alpha a_1 ,\alpha^{-1} [a_2 , a_3 ,\dots , a_{2n+4}]] ,
\end{array}
\right .
\end{equation}
Finally, using again \eqref{resn1}, we obtain 
\begin{equation}
[\alpha a_1 ,\alpha^{-1}a_2 ,\dots ,\alpha a_{2n+3}, \alpha^{-1}a_{2n+4}] = 
\alpha [a_1 ,a_2 , a_3 ,\dots , a_{2n+4}] ,
\end{equation}
which concludes the proof.
\end{proof}
 
For a given inertance $a$, we have 
\begin{equation}
Y(a )(\omega )= \di\frac{1}{ja\omega} ,
\end{equation}
so that for all $\sigma\not=0$, we obtain
\begin{equation}
\label{scaleiner}
Y(\sigma a ) (\omega ) = \sigma Y(a) (\sigma^2 \omega ) .
\end{equation}
In the same way, for a given resistance $b$, as 
\begin{equation}
Y(b)(\omega ) =\di\frac{1}{b} ,
\end{equation}
we deduce that 
\begin{equation}
\label{scaleresis}
Y(\sigma^{-1} b ) (\omega ) = \sigma Y(b) (\omega ) .
\end{equation}
As the function is independent of $j\omega$, we can also write 
\begin{equation}
\label{scaleresis2}
Y(\sigma^{-1} b ) (\omega ) = \sigma Y(b) (\sigma^2 \omega ) .
\end{equation}

Using the previous results and formula \eqref{recursiveformula}, we obtain that 
$Y_{\sigma a_1 ,\sigma^{-1} b_1}^{\sigma}$ must satisfy
\begin{equation}
\left .
\begin{array}{lll}
Y_{\sigma a_1 ,\sigma^{-1} b_1}^{\sigma} (\omega ) & = & 
[Y(\sigma^{-1} b_1),Y(\sigma a_1)^{-1}, Y(\sigma^{-2} b_1),Y(\sigma^2 a_1 )^{-1},\dots ,Y(\sigma^{-n} b_1),Y(\sigma^n a_1)^{-1},\dots ](\omega )\\
 & = & [\sigma Y(b_1) ,\sigma^{-1} Y(a_1)^{-1} ,\dots , \sigma Y(\sigma^{n-1} b_1) ,\sigma^{-1} Y(\sigma^{n-1} a_1)^{-1},\dots ] (\sigma^2 \omega ).
 \end{array}
 \right .
 \end{equation}
Then we deduce from Lemma \ref{continuedfraction} that 
\begin{equation}
\left .
\begin{array}{llll}
Y_{\sigma a_1 ,\sigma^{-1} b_1}^{\sigma} (\omega ) & = & \sigma [Y(b_1) ,Y(a_1)^{-1} ,\dots , Y(\sigma^{n-1} b_1) , Y(\sigma^{n-1} a_1)^{-1},\dots ] (\sigma^2 \omega )\\
 & = & \sigma Y_{a_1,b_1}^{\sigma} (\sigma^2 \omega ) ,
\end{array}
\right .
\end{equation}
which concludes the proof of Lemma \ref{lemmeasympfrac}. 

\ack A. Szafra\'{n}ska thanks the National Science Center for the financial support, under the research project No. 2021/05/X/ST1/00332 and J. Cresson thanks the GDR CNRS no. 2043 Géométrie différentielle et Mécanique and the fédération MARGAUx (FR 2045) for supports. 
  
\bibliographystyle{acmurl}
\bibliography{bibmodel}

\end{document}